\documentclass[11pt]{article}
\usepackage{CJK}
\usepackage{latexsym,bm}
\usepackage{xypic}
\usepackage{amsmath}
\usepackage{wasysym}
\usepackage{indentfirst}
\usepackage{amssymb}
\usepackage{dsfont}
\usepackage{amsthm}
\begin{document}
\newcommand{\fr}[2]{\frac{\;#1\;}{\;#2\;}}
\newtheorem{theorem}{Theorem}[section]
\newtheorem{lemma}{Lemma}[section]
\newtheorem{proposition}{Proposition}[section]
\newtheorem{corollary}{Corollary}[section]
\newtheorem{conjecture}{Conjecture}[section]
\newtheorem{remark}{Remark}[section]
\newtheorem{definition}{Definition}[section]
\newtheorem{notation}{Notation}[section]
\numberwithin{equation}{section}
\newcommand{\Aut}{\mathrm{Aut}\,}
\newcommand{\CSupp}{\mathrm{CSupp}\,}
\newcommand{\Supp}{\mathrm{Supp}\,}
\newcommand{\rank}{\mathrm{rank}\,}
\newcommand{\col}{\mathrm{col}\,}
\newcommand{\len}{\mathrm{len}\,}
\newcommand{\leftlen}{\mathrm{leftlen}\,}
\newcommand{\rightlen}{\mathrm{rightlen}\,}
\newcommand{\length}{\mathrm{length}\,}
\newcommand{\wt}{\mathrm{wt}\,}
\newcommand{\diff}{\mathrm{diff}\,}
\newcommand{\lcm}{\mathrm{lcm}\,}
\newcommand{\dom}{\mathrm{dom}\,}
\newcommand{\SUPP}{\mathrm{SUPP}\,}
\newcommand{\supp}{\mathrm{supp}\,}
\newcommand{\End}{\mathrm{End}\,}
\newcommand{\Hom}{\mathrm{Hom}\,}
\newcommand{\ran}{\mathrm{ran}\,}
\title{Fourier-Reflexive Partitions Induced by Poset Metric$^\ast$}
\author{Yang Xu$^1$ \,\,\,\,\,\, Haibin Kan$^2$\,\,\,\,\,\,Guangyue Han$^3$}

\maketitle

\renewcommand{\thefootnote}{\fnsymbol{footnote}}

\footnotetext{\hspace*{-12mm} \begin{tabular}{@{}r@{}p{13.4cm}@{}}
$^\ast$ & A preliminary version of this work has been presented in IEEE Symposium on Information Theory (ISIT) 2021.\\
$^1$ & Shanghai Key Laboratory of Intelligent Information Processing, School of Computer Science, Fudan University,
Shanghai 200433, China.\\
&Department of Mathematics, Faculty of Science, The University of Hong Kong, Pokfulam Road, Hong Kong, China. {E-mail:12110180008@fudan.edu.cn} \\
$^2$ & Shanghai Key Laboratory of Intelligent Information Processing, School of Computer Science, Fudan University,
Shanghai 200433, China.\\
&Fudan-Zhongan Joint Laboratory of Blockchain and Information Security, Shanghai Engineering Research Center of Blockchain, Fudan University, Shanghai 200433, China. {E-mail:hbkan@fudan.edu.cn} \\
$^3$ & Department of Mathematics, Faculty of Science, The University of Hong Kong, Pokfulam Road, Hong Kong, China. {E-mail:ghan@hku.hk} \\
\end{tabular}}

\vskip 3mm

 {\hspace*{-6mm}\bf Abstract---}\! Let $\mathbf{H}$ be the cartesian product of a family of finite abelian groups indexed by a finite set $\Omega$. A given poset (i.e., partially ordered set) $\mathbf{P}=(\Omega,\preccurlyeq_{\mathbf{P}})$ gives rise to a poset metric on $\mathbf{H}$, which further leads to a partition $\mathcal{Q}(\mathbf{H},\mathbf{P})$ of $\mathbf{H}$. We prove that if $\mathcal{Q}(\mathbf{H},\mathbf{P})$ is Fourier-reflexive, then its dual partition $\Lambda$ coincides with the partition of $\hat{\mathbf{H}}$ induced by $\mathbf{\overline{P}}$, the dual poset of $\mathbf{P}$, and moreover, $\mathbf{P}$ is necessarily hierarchical. This result establishes a conjecture proposed by Gluesing-Luerssen in \cite{4}. We also show that with some other assumptions, $\Lambda$ is finer than the partition of $\hat{\mathbf{H}}$ induced by $\mathbf{\overline{P}}$. In addition, we give some  necessary and sufficient conditions for $\mathbf{P}$ to be hierarchical, and for the case that $\mathbf{P}$ is hierarchical, we give an explicit criterion for determining whether two codewords in $\hat{\mathbf{H}}$ belong to the same block of $\Lambda$. We prove these results by relating the involved partitions with certain family of polynomials, a generalized version of which is also proposed and studied to generalize the aforementioned results.

\section{Introduction}
\setlength{\parindent}{2em}

For a nonempty set $E$, a \textit{partition} of $E$ is a collection of nonempty disjoint subsets of $E$ whose union is $E$. MacWilliams identities based on partitions of finite abelian groups have been proven by Zinoviev and Ericson in \cite{16}, and by Gluesing-Luerssen in \cite{4}. To state their results more precisely, we first introduce some notations.

For a partition $\Psi$ of $E$, we will refer to any member $A\in\Psi$ as a \textit{block} of $\Psi$, and for any $u,v\in E$, we write $u\sim_{\Psi}v$ if $u$ and $v$ belong to the same block of $\Psi$. Moreover, for any two partitions $\Psi_1$, $\Psi_2$ of $E$, $\Psi_1$ is said to be \textit{finer} than $\Psi_2$ if any block of $\Psi_1$ is contained in some block of $\Psi_2$.

Now let $H$ be a finite abelian group. The \textit{character group} of $H$, denoted by $\hat{H}$, is the set of all the group homomorphisms from $H$ to $\mathbb{C}^{*}$ (the multiplicative group of $\mathbb{C}$), equipped with multiplication
$$\text{$(\chi\psi)(a)=\chi(a)\psi(a)$ for all $\chi,\psi\in\hat{H}$,~$a\in H$.}$$
It is well known that $H\cong\hat{H}$ as groups (see [5, Chapter 5, Theorem 6.4]).

For a partition $\Gamma$ of $H$, following [4, Definition 2.1], the \textit{dual partition} of $\Gamma$, denoted by $\widehat{\Gamma}$, is the partition of $\hat{H}$ such that for any $\chi,\psi\in\hat{H}$, $\chi\sim_{\widehat{\Gamma}}\psi$ if and only if
\begin{equation}\text{$\sum_{b\in B}\chi(b)=\sum_{b\in B}\psi(b)$~for all $B\in\Gamma$}.\end{equation}
For $A\in\widehat{\Gamma}$, $B\in \Gamma$, the \textit{generalized Krawtchouk coefficient} $K_{(A,B)}$ is defined as
\begin{equation}\text{$K_{(A,B)}=\sum_{b\in B}\chi(b)$ for any chosen $\chi\in A$}.\end{equation}
For an \textit{additive code} (i.e., a subgroup) $D\subseteq H$, the dual code of $D$, denoted by $^{\bot}D$, is defined as
$$^{\bot}D=\{\chi\mid \chi\in\hat{H},~\chi(d)=1~\text{for all $d\in D$}\}.$$
It has been proven in [4, Theorem 2.7] that
\begin{equation}\text{$|^{\bot}D||D\cap B|=\sum_{A\in\widehat{\Gamma}}|^{\bot}D\cap A| K_{(A,B)}$ for all $B\in\Gamma$.}\end{equation}
The identity in (1.3) provides a general framework for recovering known or deriving new MacWilliams identities, and for many special cases, computing the generalized Krawtchouk coefficients (1.2) leads to corresponding MacWilliams identities in more explicit forms. As an example, consider $H=\mathbb{F}^{n}$ for some finite field $\mathbb{F}$ and $n\in\mathbb{Z}^{+}$, and let $\Gamma$ be the partition of $H$ induced by Hamming weight. Then, (1.3) recovers the classical MacWilliams identity relating the weight distributions of a linear code with those of its dual code (see \cite{10}). More relevant results and examples can be found in [4, Sections 2, 3 and 5]. A MacWilliams identity for $F$-partition (see [16, Definition 1]) has been established in [16, Theorem 1]. A different approach towards MacWilliams identities focusing on numerical weights due to Ravagnani can be found in \cite{14}, where the author proves a MacWilliams identity based on regular support (see [14, Theorem 29]), which also recovers many MacWilliams identities in coding theory  (see [14, Section 6]).

Now, we identify $\hat{\hat{H}}$ with $H$ via the canonical isomorphism $\rho:H\longrightarrow\hat{\hat{H}}$ defined as
$$\text{$(\rho(a))(\chi)=\chi(a)$ for all $a\in H$,~$\chi\in\hat{H}$}$$
(see [5, Chapter 5, Theorem 6.4]). Then, the \textit{bidual partition of $\Gamma$}, denoted by $\widehat{\widehat{\Gamma}}$, is the partition of $H$ such that for any $b,d\in H$, $b\sim_{\widehat{\widehat{\Gamma}}}d$ if and only if
\begin{equation}\text{$\sum_{\chi\in A}\chi(b)=\sum_{\chi\in A}\chi(d)$~for all $A\in\widehat{\Gamma}$}.\end{equation}
With the above notions, we state the definition of Fourier-reflexive partition, which has been introduced in [4, Definition 2.1].

\begin{definition}
{$\Gamma$ is said to be Fourier-reflexive if $\widehat{\widehat{\Gamma}}=\Gamma$.
}
\end{definition}

The following general result, which has been established in [4, Theorem 2.4], provides a simple and useful criterion for Fourier-reflexivity of partitions.

\begin{theorem}
{$|\Gamma|\leqslant|\widehat{\Gamma}|$, $\widehat{\widehat{\Gamma}}$ is finer than $\Gamma$, and moreover, $\Gamma$ is Fourier-reflexive if and only if $|\Gamma|=|\widehat{\Gamma}|$.
}
\end{theorem}

Fourier-reflexivity leads to some interesting and useful properties. For example, for a Fourier-reflexive partition $\Gamma$, its generalized Krawtchouk matrix $(K_{(A,B)}\mid (A,B)\in\widehat{\Gamma}\times\Gamma)$ is invertible, whose inverse is essentially determined by the generalized Krawtchouk matrix of $\widehat{\Gamma}$ (see [3, Section 2.3]). Fourier-reflexive partitions can be alternatively characterized in terms of association schemes (see [3, Theorem 2.9], [4, Section 2], [17, Theorem 1]).

In this paper, we follow [4, Section 5] to further study partitions induced by poset metric (see \cite{1}). Let $\Omega$ be a finite set and let $\mathbf{P}=(\Omega,\preccurlyeq_{\mathbf{P}})$ be a poset. For any $B\subseteq \Omega$, $B$ is said to be an \textit{ideal} of $\mathbf{P}$ if for any $b\in B$ and $a\in \Omega$, $a\preccurlyeq_{\mathbf{P}}b$ implies $a\in B$. We let $\mathcal{I}(\mathbf{P})$ denote the set of all ideals of $\mathbf{P}$. For any $B\subseteq \Omega$, we let $\max_{\mathbf{P}}(B)$ (\textit{resp}. $\min_{\mathbf{P}}(B)$) denote the set of all the maximal (\textit{resp}. minimal) elements with respect to $\preccurlyeq_{\mathbf{P}}$ in $B$, and let $\langle B\rangle_{\mathbf{P}}$ denote the ideal $\{a\mid a\in \Omega~s.t.~\exists~b\in B,~a\preccurlyeq_{\mathbf{P}}b\}$. The \textit{dual poset of $\mathbf{P}$} will be denoted by $\mathbf{\overline{P}}=(\Omega,\preccurlyeq_{\mathbf{\overline{P}}})$, where $$\text{$u\preccurlyeq_{\mathbf{\overline{P}}} v\Longleftrightarrow v\preccurlyeq_{\mathbf{P}}u$ for all $u,v\in \Omega$.}$$
For any $y\in\Omega$, we let $\len(y)$ denote the largest cardinality of a chain in $\mathbf{P}$ containing $y$ as its greatest element. The following notion of \textit{hierarchical poset} has been introduced in \cite{4,7,9,11,13}.

\begin{definition}
{$\mathbf{P}$ is said to be hierarchical if for any $u,v\in \Omega$ with $\len(u)+1\leqslant\len(v)$, it holds that $u\preccurlyeq_{\mathbf{P}} v$.
}
\end{definition}

Let $\left(H_{i}\mid i\in \Omega\right)$ be a family of finite abelian groups and let $\mathbf{H}\triangleq\prod_{i\in \Omega}H_{i}$. For any \textit{codeword} $\beta\in \mathbf{H}$, we define $\supp(\beta)$ as
$$\supp(\beta)=\{i\mid i\in \Omega,~\beta_{(i)}\neq1_{H_{i}}\}.$$
Following \cite{1,4,12}, for any codeword $\beta$, the $\mathbf{P}$-weight of $\beta$, denoted by $\wt_{\mathbf{P}}(\beta)$, is defined as \begin{equation}\wt_{\mathbf{P}}(\beta)=|\langle\supp(\beta)\rangle_{\mathbf{P}}|.\end{equation}
$\mathbf{P}$-weight naturally gives rise to a partition of $\mathbf{H}$, as detailed in the following notation.

\begin{notation}
Let $\mathcal{Q}(\mathbf{H},\mathbf{P})$ denote the partition of $\mathbf{H}$ such that for any $\beta,\gamma\in\mathbf{H}$, $\beta$ and $\gamma$ belong to the same block of $\mathcal{Q}(\mathbf{H},\mathbf{P})$ if and only if $\wt_{\mathbf{P}}(\beta)=\wt_{\mathbf{P}}(\gamma)$.
\end{notation}

The following conjecture has been proposed in [4, Section 5] (see the paragraph right after [4, Theorem 5.4]).

\begin{conjecture}
{If $|H_{i}|\geqslant2$ for all $i\in\Omega$ and $\mathcal{Q}(\mathbf{H},\mathbf{P})$ is Fourier-reflexive, then $\mathbf{P}$ is hierarchical.
}
\end{conjecture}

Now we identify $\hat{\mathbf{H}}$ with $\prod_{i\in \Omega}\hat{H_{i}}$ by viewing each $\alpha\in\prod_{i\in \Omega}\hat{H_{i}}$ as an element of $\hat{\mathbf{H}}$ via
\begin{equation}\text{$\alpha(\beta)=\prod_{i\in\Omega}\alpha_{(i)}\left(\beta_{(i)}\right)$ for all $\beta\in \mathbf{H}$}\end{equation}
(see [4, Section 2]). With this identification, for any $\alpha\in\hat{\mathbf{H}}$, $\supp(\alpha)$ is the set $\{i\mid i\in \Omega,~\alpha_{(i)}\neq1_{\hat{H_{i}}}\}$. Furthermore, we can also consider poset weight for codewords in $\hat{\mathbf{H}}$ as in (1.5) and hence Notation 1.1 also applies. For a poset $\mathbf{P}=(\Omega,\preccurlyeq_{\mathbf{P}})$, MacWilliams identities relating the $\mathbf{P}$-weight defined over $\mathbf{H}$ with the $\mathbf{\overline{P}}$-weight defined over $\hat{\mathbf{H}}$ have been established in \cite{4,7,13}, and the following conjecture has been proposed in [4, Section 5]:
\begin{conjecture}
{If $|H_{i}|\geqslant2$ for all $i\in\Omega$, then the dual partition of $\mathcal{Q}(\mathbf{H},\mathbf{P})$ is finer than $\mathcal{Q}(\hat{\mathbf{H}},\mathbf{\overline{P}})$.
}
\end{conjecture}

We note that with the help of Theorem 1.1 and [4, Theorem 5.4], one can verify that Conjecture 1.2 implies Conjecture 1.1.

In this paper, we prove Conjecture 1.1, and with the additional assumption that either $|H_{i}|=|H_{l}|\geqslant3$ for all $i,l\in\Omega$ or $\mathbf{P}$ is hierarchical, we show that Conjecture 1.2 holds true. We also establish some sufficient and necessary conditions for $\mathbf{P}$ to be hierarchical, and for the case that $\mathbf{P}$ is hierarchical, we give an explicit criterion for determining whether two codewords in $\hat{\mathbf{H}}$ belong to the same block of $\Lambda$, where $\Lambda$ denotes the dual partition of $\mathcal{Q}(\mathbf{H},\mathbf{P})$. We use a polynomial approach to prove these results and some of their generalizations, as detailed in the next two paragraphs.

In Section 2, using the Krawtchouk coefficients of $\mathcal{Q}(\mathbf{H},\mathbf{P})$, we first relate each $\alpha\in\hat{\mathbf{H}}$ with a polynomial $F(\alpha)$ (Notation 2.1). A key observation underpinning our approach is that any two codewords $\alpha,\gamma\in\hat{\mathbf{H}}$ belong to the same block of $\Lambda$ if and only if $F(\alpha)=F(\gamma)$ (Lemma 2.1). Based on this observation, we study the partition $\Lambda$ via the family of polynomials $(F(\alpha)\mid\alpha\in\hat{\mathbf{H}})$. In Section 2.1, we compute $F(\alpha)$ (Theorem 2.1, Propositions 2.1 and 2.2). In Section 2.2, by examining the leading coefficients, we prove Conjecture 1.2 with the additional assumption that $|H_{i}|=|H_{l}|\geqslant3$ for all $i,l\in\Omega$ (Theorem 2.2). In Section 2.3, we give a necessary and sufficient condition for $\mathbf{P}$ to be hierarchical in terms of $\overline{\mathbf{P}}$-weight defined over $\hat{\mathbf{H}}$ and the degree of $F(\alpha)$ (Theorem 2.3). In Section 2.4, we state Theorem 2.4 which establishes Conjecture 1.1. In Section 2.5, for the case that $\mathbf{P}$ is hierarchical, we give a necessary and sufficient condition for two codewords $\alpha,\gamma\in\hat{\mathbf{H}}$ belong to the same block of $\Lambda$ in terms of $\langle\supp(\alpha)\rangle_{\overline{\mathbf{P}}}$ and $\langle\supp(\gamma)\rangle_{\overline{\mathbf{P}}}$, and as a consequence, we prove Conjecture 1.2 with the additional assumption that $\mathbf{P}$ is hierarchical (Theorem 2.5, Corollary 2.2).

In Section 3, with respect to the same poset $\mathbf{P}=(\Omega,\preccurlyeq_{\mathbf{P}})$, we propose and study a generalized version of $F(\alpha)$, which we denote by $\pi(Y,D)$, where $D\subseteq Y\subseteq\Omega$ (see Equation (3.2)). Some basic properties of such polynomials are given in Section 3.1 (Theorems 3.1 and 3.2). In Section 3.2, we study conditions for $\mathbf{P}$ to be hierarchical in terms of $\pi(\cdot,\cdot)$. We compute $\pi(\Omega,D)$ for the case that $\mathbf{P}$ is hierarchical (Lemma 3.2), and give generalizations of Theorem 2.3, Theorem 2.5 and [4, Theorems 5.4 and 5.5] (Proposition 3.2, Proposition 3.4, Theorem 3.3). In Section 3.3, we prove that if the base field of $\pi(Y,D)$ is set to be $\mathbb{R}$, then for $D,A\in\mathcal{I}(\overline{\mathbf{P}})$ with $A\subsetneqq D$, $\pi(\Omega,D)$ is ``smaller'' than $\pi(\Omega,A)$ with respect to certain total order defined on $\mathbb{R}[x]$ (Theorem 3.4), and as an application, we prove a generalization of Theorem 2.4 with the help of the set $\{\pi(\Omega,M)\mid M\in\mathcal{I}(\overline{\mathbf{P}})\}$ (Theorem 3.5). Finally, in Section 3.4, we prove Theorems 2.4 and 2.5 as consequences of Theorem 3.5 and Proposition 3.4, respectively.

We remark that using the method introduced by Oh in \cite{11} and by Machado, Pinheiro and Firer in \cite{9}, together with the result established by Zinoviev and Ericson in \cite{17}, Conjecture 1.1 can be alternatively proved in terms of association schemes. In appendix A, we give such a proof, which closely follows the spirits in the proofs of [11, Theorem 2.9] and [9, Theorem 3, part $\mathfrak{B}_3$]. In terms of establishing Conjecture 1.1, the association scheme approach is more direct, and the proof is shorter. On the other hand, as presented in Section 2, other than establishing Conjecture 1.1, our polynomial approach can also be used to establish other results which do not require the Fourier-reflexivity of $\mathcal{Q}(\mathbf{H},\mathbf{P})$. Also, as presented in Section 3, by studying $\pi(\cdot,\cdot)$, which is a generalized version of $F(\alpha)$, we are able to prove generalizations of some coding-theoretic results presented in Section 2, as well as to prove results which are perhaps of interest in their own right (e.g., Theorems 3.2 and 3.4).

\section{The partition $\mathcal{Q}(\mathbf{H},\mathbf{P})$}

\setlength{\parindent}{2em}

Throughout the paper, for any $a,b\in\mathbb{Z}$, we let $[a,b]$ denote the set $\{i\mid i\in\mathbb{Z},~a\leqslant i\leqslant b\}$. Note that if $a\geqslant b+1$, then $[a,b]=\emptyset$.

In this section, we let $\Omega$ be a finite set with $|\Omega|=n$, and let $\left(H_{i}\mid i\in \Omega\right)$ be a family of finite abelian groups. For any $i\in\Omega$, we let $h_{i}\triangleq|H_{i}|$. As in Section 1, we consider $\mathbf{H}=\prod_{i\in \Omega}H_{i}$ and we will use the identification $\hat{\mathbf{H}}=\prod_{i\in \Omega}\hat{H_{i}}$ as in (1.6).

From now on until the end of this section, we fix a poset $\mathbf{P}=(\Omega,\preccurlyeq_{\mathbf{P}})$, and we consider the partition $\mathcal{Q}(\mathbf{H},\mathbf{P})$. Furthermore, we let $\Lambda$ denote the dual partition of $\mathcal{Q}(\mathbf{H},\mathbf{P})$.

We introduce the following notation, where we relate each $\alpha\in\hat{\mathbf{H}}$ with a polynomial.

\begin{notation}
{For any $\alpha\in\hat{\mathbf{H}}$, let $F(\alpha)\triangleq\sum_{l=0}^{n}\left(\sum_{(\beta\in \mathbf{H},\wt_{\mathbf{P}}(\beta)=l)}\alpha(\beta)\right)x^{l}$.
}
\end{notation}

\subsection{The polynomial $F(\alpha)$}
\setlength{\parindent}{2em}
First of all, with the help of Notation 2.1, (1.1) can be immediately reformulated for $\mathcal{Q}(\mathbf{H},\mathbf{P})$, as detailed in the following lemma.

\setlength{\parindent}{0em}
\begin{lemma}
{For any $\alpha,\gamma\in\hat{\mathbf{H}}$, $\alpha\sim_{\Lambda}\gamma\Longleftrightarrow F(\alpha)=F(\gamma)$.
}
\end{lemma}

\setlength{\parindent}{0em}
\begin{remark}
{Lemma 2.1 is the key observation underpinning our approach in this paper. Based on Lemma 2.1, we will study the partition $\Lambda$ via the family of polynomials $(F(\alpha)\mid\alpha\in\hat{\mathbf{H}})$.
}
\end{remark}

\setlength{\parindent}{2em}
The following lemma is an immediate corollary of [6, Proposition 1.1 and Lemma 1.2] (also see [12, Proposition 7]).

\setlength{\parindent}{0em}
\begin{lemma}
{\bf{(1)}}\,\,For any $U\in\mathcal{I}(\mathbf{P})$ with $U\subsetneqq \Omega$, there exists $V\in\mathcal{I}(\mathbf{P})$ such that $U\subseteq V$ and $|V|=|U|+1$.

{\bf{(2)}}\,\,For any $V\in\mathcal{I}(\mathbf{P})$ with $V\neq\emptyset$, there exists $U\in\mathcal{I}(\mathbf{P})$ such that $U\subseteq V$ and $|U|=|V|-1$.

{\bf{(3)}}\,\,$\mathcal{I}(\mathbf{\overline{P}})=\{\Omega-A\mid A\in \mathcal{I}(\mathbf{P})\}$.
\end{lemma}

\setlength{\parindent}{2em}
Using Lemma 2.2 and Theorem 1.1, we establish the following corollary.

\setlength{\parindent}{0em}
\begin{corollary}
If $h_{i}\geqslant2$ for all $i\in\Omega$, then $|\mathcal{Q}(\mathbf{H},\mathbf{P})|=n+1$, $|\Lambda|\geqslant n+1$, and moreover, $\mathcal{Q}(\mathbf{H},\mathbf{P})$ is Fourier-reflexive $\Longleftrightarrow$ $|\Lambda|=n+1$.
\end{corollary}

\begin{proof}
Since $h_{i}\geqslant2$ for all $i\in\Omega$, we deduce that for any $I\subseteq \Omega$, there exists $\beta\in\mathbf{H}$ with $\supp(\beta)=I$. This fact, together with Lemma 2.2, implies that $|\mathcal{Q}(\mathbf{H},\mathbf{P})|=|\Omega|+1=n+1$. The rest immediately follows from Theorem 1.1.
\end{proof}

\setlength{\parindent}{2em}
Now we compute $F(\alpha)$, and we begin with the following proposition.

\begin{proposition}
{Let $\alpha\in \hat{\mathbf{H}}$ with $\langle \supp(\alpha)\rangle_{\mathbf{\overline{P}}}=D$, and let $I\in\mathcal{I}(\mathbf{P})$. Then, we have
\begin{eqnarray*}
\begin{split}
&\sum_{(\beta\in\mathbf{H},\langle \supp(\beta)\rangle_{\mathbf{P}}=I)}\alpha(\beta)\\
&=\begin{cases}
(-1)^{|I\cap D|}\left(\prod_{i\in I-\max_{\mathbf{P}}(I)}h_{i}\right)\left(\prod_{i\in \max_{\mathbf{P}}(I)-D}(h_{i}-1)\right),&I\cap D\subseteq \max_{\mathbf{P}}(I);\\
0,&I\cap D\not\subseteq \max_{\mathbf{P}}(I).
\end{cases}
\end{split}
\end{eqnarray*}
}
\end{proposition}

\begin{proof}
Throughout the proof, we let $J=\max_{\mathbf{P}}(I)$. Note that for any $\beta\in\mathbf{H}$, it holds that
$$\langle \supp(\beta)\rangle_{\mathbf{P}}=I\Longleftrightarrow J\subseteq \supp(\beta) \subseteq I,$$
which, together with the principle of inclusion-exclusion (see [15, Theorem 2.1.1]) and some straightforward computation, implies that
\begin{eqnarray*}
\begin{split}
\sum_{(\beta\in\mathbf{H},\langle \supp(\beta)\rangle_{\mathbf{P}}=I)}\alpha(\beta)&=\text{$\sum_{(U,J\subseteq U \subseteq I)}\sum_{(\beta\in\mathbf{H}, \supp(\beta)=U)}\alpha(\beta)$}\\
&=\text{$\sum_{(A,I-J\subseteq A\subseteq I)}(-1)^{|I|-|A|}\left(\sum_{(\beta\in \mathbf{H},\supp(\beta)\subseteq A)}\alpha(\beta)\right)$.}
\end{split}
\end{eqnarray*}
Now, consider an arbitrary $A\subseteq \Omega$ with $I-J\subseteq A\subseteq I$. By the orthogonality relation (see [8, Theorem 5.4]), we conclude that
$$\text{$\sum_{(\beta\in \mathbf{H},\supp(\beta)\subseteq A)}\alpha(\beta)=\begin{cases}
\,\prod_{i\in A}h_{i},& \supp(\alpha)\subseteq \Omega-A;\\
\,0,& \supp(\alpha)\not\subseteq \Omega-A.
\end{cases}$}
$$
Since $I\in\mathcal{I}(\mathbf{P})$, $J=\max_{\mathbf{P}}(I)$, we have $A\in\mathcal{I}(\mathbf{P})$, and hence $\Omega-A\in\mathcal{I}(\mathbf{\overline{P}})$ by (3) of Lemma 2.2. Noticing that $\langle \supp(\alpha)\rangle_{\mathbf{\overline{P}}}=D$, we deduce that
$$\supp(\alpha)\subseteq \Omega-A\Longleftrightarrow D\subseteq \Omega-A\Longleftrightarrow A\subseteq \Omega-D.$$
By the above discussion, we further deduce that
\begin{eqnarray*}
\begin{split}
\sum_{(\beta\in\mathbf{H},\langle \supp(\beta)\rangle_{\mathbf{P}}=I)}\alpha(\beta)&=\sum_{(A,I-J\subseteq A\subseteq I,A\subseteq\Omega-D)}(-1)^{|I|-|A|}\left(\prod_{i\in A}h_{i}\right)\\
&=\sum_{(A,I-J\subseteq A\subseteq I-D)}(-1)^{|I|-|A|}\left(\prod_{i\in A}h_{i}\right).
\end{split}
\end{eqnarray*}
Now treating $I\cap D\subseteq J$ and $I\cap D\not\subseteq J$ separately, the proposition follows from some straightforward computation which we omit.
\end{proof}

\setlength{\parindent}{2em}
The following lemma is straightforward to verify.

\begin{lemma}
{For $D\in\mathcal{I}(\mathbf{\overline{P}})$ and $I\in\mathcal{I}(\mathbf{P})$, we have
$$\mbox{$I\cap D\subseteq \max_{\mathbf{P}}(I)\Longleftrightarrow I\cap D\subseteq \min_{\mathbf{P}}(D)\Longleftrightarrow I\subseteq (\Omega-D)\cup\min_{\mathbf{P}}(D)$}.$$
}
\end{lemma}

\setlength{\parindent}{2em}
A combination of Proposition 2.1 and Lemma 2.3 yields the following theorem.

\begin{theorem}
{Let $\alpha\in \hat{\mathbf{H}}$ with $\langle \supp(\alpha)\rangle_{\mathbf{\overline{P}}}=D$. Then, $F(\alpha)$ is equal to
$$\text{\hspace*{-12mm}$\sum_{(I\in\mathcal{I}(\mathbf{P}),I\subseteq (\Omega-D)\cup\min_{\mathbf{P}}(D))}(-1)^{|I\cap D|}\left(\prod_{i\in I-\max_{\mathbf{P}}(I)}h_{i}\right)\left(\prod_{i\in \max_{\mathbf{P}}(I)-\min_{\mathbf{P}}(D)}(h_{i}-1)\right) x^{|I|}$.}$$
}
\end{theorem}

\setlength{\parindent}{2em}
As a corollary of Theorem 2.1, we have the following proposition.

\setlength{\parindent}{0em}
\begin{proposition}
{Assume that $h_{i}\geqslant2$ for all $i\in\Omega$. Let $\alpha\in \hat{\mathbf{H}}$ with $\langle \supp(\alpha)\rangle_{\mathbf{\overline{P}}}=D$, and let $X=(\Omega-D)\cup\min_{\mathbf{P}}(D)$. Then, we have
$$\mbox{$\deg(F(\alpha))=|X|=|\Omega|-|D|+|\min_{\mathbf{P}}(D)|$},$$
and moreover, the leading coefficient of $F(\alpha)$ is equal to
$$(-1)^{|\min_{\mathbf{P}}(D)|}\left(\prod_{i\in X-\max_{\mathbf{P}}(X)}h_{i}\right)\left(\prod_{i\in \max_{\mathbf{P}}(X)-\min_{\mathbf{P}}(D)}(h_{i}-1)\right).$$
}
\end{proposition}

\setlength{\parindent}{2em}
Proposition 2.2 provides a necessary condition for two codewords of $\hat{\mathbf{H}}$ belong to the same block of $\Lambda$, as detailed in the following proposition.

\setlength{\parindent}{0em}
\begin{proposition}
{Assume that $h_{i}\geqslant2$ for all $i\in\Omega$, and let $\alpha,\gamma\in \hat{\mathbf{H}}$ such that $\alpha\sim_{\Lambda}\gamma$. Then, it holds that
$$\mbox{$|D|-|\min_{\mathbf{P}}(D)|=|B|-|\min_{\mathbf{P}}(B)|$},$$
where $D=\langle \supp(\alpha)\rangle_{\mathbf{\overline{P}}}$, $B=\langle \supp(\gamma)\rangle_{\mathbf{\overline{P}}}$.
}
\end{proposition}

\begin{proof}
By Lemma 2.1, we have $F(\alpha)=F(\gamma)$, and hence $\deg(F(\alpha))=\deg(F(\gamma))$. Now the result immediately follows from Proposition 2.2.
\end{proof}

\subsection{Establishing Conjecture 1.2 with an additional assumption}
\setlength{\parindent}{2em}
In this subsection, we prove Conjecture 1.2 with the additional assumption that $h_{i}=h_{l}\geqslant3$ for all $i,l\in \Omega$, and we begin with the following lemma.

\setlength{\parindent}{0em}
\begin{lemma}
{Assume $h_{i}=h_{l}\geqslant3$ for all $i,l\in \Omega$. Let $\alpha,\gamma\in \hat{\mathbf{H}}$ and let $X=(\Omega-D)\cup\min_{\mathbf{P}}(D)$, $V=(\Omega-B)\cup\min_{\mathbf{P}}(B)$, where $D=\langle \supp(\alpha)\rangle_{\mathbf{\overline{P}}}$, $B=\langle \supp(\gamma)\rangle_{\mathbf{\overline{P}}}$. Then, the following two conditions are equivalent to each other:

{\bf{(1)}}\,\,The leading coefficients of $F(\alpha)$ and $F(\gamma)$ have the same absolute value;

{\bf{(2)}}\,\,$|X|-|\max_{\mathbf{P}}(X)|=|V|-|\max_{\mathbf{P}}(V)|$, $|D|=|B|$.
}
\end{lemma}

\begin{proof}
Let $a:=h_{i}$ for all $i\in \Omega$. Noticing that $D,B\in\mathcal{I}(\mathbf{\overline{P}})$, we deduce that $\min_{\mathbf{P}}(D)\subseteq\max_{\mathbf{P}}(X)$, $\min_{\mathbf{P}}(B)\subseteq\max_{\mathbf{P}}(V)$. By Proposition 2.2, the absolute value of the leading coefficient of $F(\alpha)$ is equal to
$$a^{|X|-|\max_{\mathbf{P}}(X)|}(a-1)^{|\max_{\mathbf{P}}(X)|-|\min_{\mathbf{P}}(D)|},$$
and the absolute value of the leading coefficient of $F(\gamma)$ is equal to
$$a^{|V|-|\max_{\mathbf{P}}(V)|}(a-1)^{|\max_{\mathbf{P}}(V)|-|\min_{\mathbf{P}}(B)|}.$$
By $a\geqslant3$ and $\gcd(a,a-1)=1$, we deduce that (1) holds true if and only if the following two conditions hold:

$(i)$\,\,$|X|-|\max_{\mathbf{P}}(X)|=|V|-|\max_{\mathbf{P}}(V)|$.

$(ii)$\,\,$|\max_{\mathbf{P}}(X)|-|\min_{\mathbf{P}}(D)|=|\max_{\mathbf{P}}(V)|-|\min_{\mathbf{P}}(B)|$.

Moreover, it can be readily verified that $(i)$ and $(ii)$ hold true if and only if both $(i)$ and the following condition holds:
$$\mbox{$|X|-|\min_{\mathbf{P}}(D)|=|V|-|\min_{\mathbf{P}}(B)|$}.$$

Noticing that $|X|-|\min_{\mathbf{P}}(D)|=|\Omega|-|D|$, $|V|-|\min_{\mathbf{P}}(B)|=|\Omega|-|B|$, we deduce that $|X|-|\min_{\mathbf{P}}(D)|=|V|-|\min_{\mathbf{P}}(B)|\Longleftrightarrow|D|=|B|$, which immediately implies the lemma.
\end{proof}

\setlength{\parindent}{2em}
Now we prove Conjecture 1.2 with the aforementioned additional assumption.

\setlength{\parindent}{0em}
\begin{theorem}
{If $h_{i}=h_{l}\geqslant3$ for all $i,l\in \Omega$, then $\Lambda$ is finer than $\mathcal{Q}(\hat{\mathbf{H}},\mathbf{\overline{P}})$.
}
\end{theorem}

\begin{proof}
Let $\alpha,\gamma\in \hat{\mathbf{H}}$ such that $\alpha\sim_{\Lambda}\gamma$. By Lemma 2.1, we have $F(\alpha)=F(\gamma)$, which, together with Lemma 2.4, further implies that $|\langle \supp(\alpha)\rangle_{\mathbf{\overline{P}}}|=|\langle \supp(\gamma)\rangle_{\mathbf{\overline{P}}}|$, and hence $\wt_{\mathbf{\overline{P}}}(\alpha)=\wt_{\mathbf{\overline{P}}}(\gamma)$. The above discussion establishes the fact that $\Lambda$ is finer than $\mathcal{Q}(\hat{\mathbf{H}},\mathbf{\overline{P}})$, completing the proof.
\end{proof}

\subsection{A necessary and sufficient condition for $\mathbf{P}$ to be hierarchical}
\setlength{\parindent}{2em}
The main result of this subsection is the following theorem, in which we use the polynomial $F(\alpha)$ to give a necessary and sufficient condition for $\mathbf{P}$ to be hierarchical.

\setlength{\parindent}{0em}
\begin{theorem}
{Assume that $h_{i}\geqslant2$ for all $i\in\Omega$. Then, the following two conditions are equivalent to each other:

{\bf{(1)}}\,\,$\mathbf{P}$ is hierarchical;

{\bf{(2)}}\,\,For any $\alpha,\gamma\in\hat{\mathbf{H}}$, $\wt_{\mathbf{\overline{P}}}(\alpha)=\wt_{\mathbf{\overline{P}}}(\gamma)\Longrightarrow\deg(F(\alpha))=\deg(F(\gamma))$.
}
\end{theorem}

\setlength{\parindent}{2em}
We begin with the following lemma, which is straightforward to verify by using Definition 1.2.
\setlength{\parindent}{0em}
\begin{lemma}
{Let $m$ denote the largest cardinality of a chain in $\mathbf{P}$, and for any $j\in[1,m]$, let $W_{j}=\{u\mid u\in\Omega,~\len(u)=j\}$. Fix $D\subseteq\Omega$. Then, there uniquely exists $r\in[1,m]$ such that $D\subseteq \bigcup_{j=r}^{m}W_{j}$ and $D\not\subseteq \bigcup_{j=p}^{m}W_{j}$ for all $p\in[r+1,m]$. Moreover, if $\mathbf{P}$ is hierarchical and $D\in\mathcal{I}(\mathbf{\overline{P}})$, then we have $D=(D\cap W_{r})\cup(\bigcup_{j=r+1}^{m}W_{j})$ and $D\cap W_{r}=\min_{\mathbf{P}}(D)$.
}
\end{lemma}

\setlength{\parindent}{2em}
We also need the following lemma. Recall that for any $Y\subseteq \Omega$, $Y$ is said to be an \emph{anti-chain in $\mathbf{P}$} if for any $u,v\in Y$, $u\preccurlyeq_{\mathbf{P}}v$ implies $u=v$.

\setlength{\parindent}{0em}
\begin{lemma}
{The following five statements are equivalent to each other:

{\bf{(1)}}\,\,$\mathbf{P}$ is hierarchical;

{\bf{(2)}}\,\,For any $D,B\in\mathcal{I}(\mathbf{\overline{P}})$, $|D|=|B|\Longrightarrow|\min_{\mathbf{P}}(D)|=|\min_{\mathbf{P}}(B)|$;

{\bf{(3)}}\,\,For any $A\subseteq\Omega$ such that $a\preccurlyeq_{\mathbf{P}}b$ for all $a\in A$, $b\in \Omega-A$, it holds that $c\preccurlyeq_{\mathbf{P}}d$ for all $c\in A-\max_{\mathbf{P}}(A)$, $d\in \max_{\mathbf{P}}(A)$;

{\bf{(4)}}\,\,For any $I,J\in\mathcal{I}(\mathbf{P})$, $|I|=|J|\Longrightarrow|\max_{\mathbf{P}}(I)|=|\max_{\mathbf{P}}(J)|$;

{\bf{(5)}}\,\,For any $B\subseteq\Omega$ such that $a\preccurlyeq_{\mathbf{P}}b$ for all $a\in \Omega-B$, $b\in B$, it holds that $c\preccurlyeq_{\mathbf{P}}d$ for all $c\in \min_{\mathbf{P}}(B)$, $d\in B-\min_{\mathbf{P}}(B)$.
}
\end{lemma}

\begin{proof}
$(1)\Longrightarrow(2)$\,\,This immediately follows from Lemma 2.5.

$(2)\Longrightarrow(3)$\,\,Assume that (2) holds. Fixing $A\subseteq\Omega$ such that $a\preccurlyeq_{\mathbf{P}}b$ for all $a\in A$, $b\in \Omega-A$, we will show
\begin{equation}\mbox{$c\preccurlyeq_{\mathbf{P}}d$ for all $c\in A-\max_{\mathbf{P}}(A)$, $d\in \max_{\mathbf{P}}(A)$},\end{equation}
which immediately implies (3). If $A=\emptyset$, then (2.1) trivially holds. Therefore in the following, we assume $A\neq\emptyset$. Let $B=(\Omega-A)\cup\max_{\mathbf{P}}(A)$. Then, it can be readily verified that $B\in\mathcal{I}(\mathbf{\overline{P}})$. Since $\max_{\mathbf{P}}(A)\neq\emptyset$, we deduce that $\min_{\mathbf{P}}(B)=\max_{\mathbf{P}}(A)$. Now, by way of contradiction, we assume that (2.1) does not hold. Then, we can choose $u\in\max_{\mathbf{P}}(A-\max_{\mathbf{P}}(A))$ and $y\in \max_{\mathbf{P}}(A)$ such that $u\not\preccurlyeq_{\mathbf{P}}y$. Let $D=(\Omega-A)\cup(\max_{\mathbf{P}}(A)-\{y\})\cup\{u\}$. It can then be readily verified that $D\in\mathcal{I}(\mathbf{\overline{P}})$, $|D|=|B|$. Since $u\in A$, we can choose $v\in \max_{\mathbf{P}}(A)$ such that $u\preccurlyeq_{\mathbf{P}}v$. Furthermore, it can be readily verified that $\min_{\mathbf{P}}(D)\subseteq\left(\max_{\mathbf{P}}(A)-\{y,v\}\right)\cup\{u\}$. Noticing that $y\neq v$, we deduce that $|\min_{\mathbf{P}}(D)|\leqslant|\max_{\mathbf{P}}(A)|-1=|\min_{\mathbf{P}}(B)|-1$, a contradiction to (2), as desired.

$(3)\Longrightarrow(1)$\,\,Assume that (3) holds. First of all, there exist $s\in\mathbb{Z}^{+}$ and a tuple of nonempty sets $(A_1,\dots,A_{s})$ such that

\hspace{6mm}\,\,$(i)$\,\,$\Omega=\bigcup_{j=1}^{s}A_{j}$;

\hspace{6mm}\,\,$(ii)$\,\,For any $j,t\in[1,s]$ with $j\neq t$, it holds that $A_{j}\cap A_{t}=\emptyset$;

\hspace{6mm}\,\,$(iii)$\,\,For any $t\in[1,s]$, it holds that $A_{t}=\max_{\mathbf{P}}\left(\bigcup_{j=1}^{t}A_{j}\right)$.

Now using (3) and an induction argument, it can be readily verified that
\begin{equation}\mbox{$\forall~j,t\in[1,s]$ s.t. $j+1\leqslant t$, we have $a\preccurlyeq_{\mathbf{P}}b$ for all $a\in A_{j}$, $b\in A_{t}$}.\end{equation}
For any $t\in[1,s]$, by $(iii)$, we deduce that $A_{t}$ is an anti-chain in $\mathbf{P}$, which, together with (2.2), further implies that $A_{t}=\{u\mid u\in\Omega,~\len(u)=t\}$. Now we apply Definition 1.2 and reach $\mathbf{P}$ is hierarchical, proving (1).

Finally, based on the following three facts:

\hspace{6mm}\,\,$(iv)$\,\,$\mathbf{P}$ is hierarchical if and only if $\overline{\mathbf{P}}$ is hierarchical;

\hspace{6mm}\,\,$(v)$\,\,$\overline{\overline{\mathbf{P}}}=\mathbf{P}$;

\hspace{6mm}\,\,$(vi)$\,\,$\min_{\mathbf{P}}(B)=\max_{\overline{\mathbf{P}}}(B)$, $\max_{\mathbf{P}}(B)=\min_{\overline{\mathbf{P}}}(B)$ for all $B\subseteq\Omega$;

$(1)\Longleftrightarrow(4)$ and $(1)\Longleftrightarrow(5)$ follow from applying $(1)\Longleftrightarrow(2)$ and $(1)\Longleftrightarrow(3)$ to $\overline{\mathbf{P}}$, respectively.
\end{proof}

\setlength{\parindent}{2em}
We are now ready to prove Theorem 2.3.\\

\setlength{\parindent}{0em}
{\textbf{Proof of Theorem 2.3:}} $(1)\Longrightarrow(2)$\,\,We fix $\alpha,\gamma\in\hat{\mathbf{H}}$ such that $\wt_{\mathbf{\overline{P}}}(\alpha)=\wt_{\mathbf{\overline{P}}}(\gamma)$, and let $D=\langle \supp(\alpha)\rangle_{\mathbf{\overline{P}}}$, $B=\langle \supp(\gamma)\rangle_{\mathbf{\overline{P}}}$. Thus we have $|D|=|B|$. Since $\mathbf{P}$ is hierarchical, by Lemma 2.6, we deduce that $|\min_{\mathbf{P}}(D)|=|\min_{\mathbf{P}}(B)|$. Now we apply Proposition 2.2 and reach $\deg(F(\alpha))=\deg(F(\gamma))$, proving (2).

$(2)\Longrightarrow(1)$\,\,We fix $D,B\in\mathcal{I}(\mathbf{\overline{P}})$ such that $|D|=|B|$. Since $h_{i}\geqslant2$ for all $i\in\Omega$, we can choose $\alpha,\gamma\in\hat{\mathbf{H}}$ such that $\langle \supp(\alpha)\rangle_{\mathbf{\overline{P}}}=D$, $\langle \supp(\gamma)\rangle_{\mathbf{\overline{P}}}=B$. Since $|D|=|B|$, we deduce that $\wt_{\mathbf{\overline{P}}}(\alpha)=\wt_{\mathbf{\overline{P}}}(\gamma)$, and hence $\deg(F(\alpha))=\deg(F(\gamma))$. Now we apply Proposition 2.2 and reach $|\min_{\mathbf{P}}(D)|=|\min_{\mathbf{P}}(B)|$. Using $(2)\Longrightarrow(1)$ of Lemma 2.6, we conclude that $\mathbf{P}$ is hierarchical, proving (1).

\subsection{Establishing Conjecture 1.1}

\setlength{\parindent}{2em}
We state the following theorem, which we will prove in Section 3.

\setlength{\parindent}{0em}
\begin{theorem}
{If $h_{i}\geqslant2$ for all $i\in\Omega$, then the following three conditions are equivalent to each other:

{\bf{(1)}}\,\,$\mathcal{Q}(\mathbf{H},\mathbf{P})$ is Fourier-reflexive;

{\bf{(2)}}\,\,$\Lambda=\mathcal{Q}(\hat{\mathbf{H}},\mathbf{\overline{P}})$;

{\bf{(3)}}\,\,$\mathbf{P}$ is hierarchical, and moreover, for any $u,v\in \Omega$ with $\len(u)=\len(v)$, it holds that $h_{u}=h_{v}$.
}
\end{theorem}

\setlength{\parindent}{2em}
It is straightforward to see that Theorem 2.4 implies Conjecture 1.1. We note that $(2)\Longleftrightarrow(3)$ of Theorem 2.4 has already been established in [4, Theorems 5.4 and 5.5]. We will prove a generalization of Theorem 2.4 in Section 3.3, and then prove Theorem 2.4 in Section 3.4.

\subsection{Characterizing $\Lambda$ for the case that $\mathbf{P}$ is hierarchical}

\setlength{\parindent}{2em}
If $\mathbf{P}$ is hierarchical, then by computing $F(\alpha)$ more explicitly, we are able to establish a relatively more explicit criterion for determining whether two codewords of $\hat{\mathbf{H}}$ belong to the same block of $\Lambda$, as detailed in the following theorem.

\setlength{\parindent}{0em}
\begin{theorem}
{Assume that $\mathbf{P}$ is hierarchical. Let $\alpha,\gamma\in\hat{\mathbf{H}}$ with
$$\langle \supp(\alpha)\rangle_{\mathbf{\overline{P}}}=D,~\langle \supp(\gamma)\rangle_{\mathbf{\overline{P}}}=B.$$
Then, the following two conditions are equivalent to each other:

{\bf{(1)}}\,\,$\alpha\sim_{\Lambda}\gamma$;

{\bf{(2)}}\,\,There exists a bijection $\varepsilon:B\longrightarrow D$ such that $h_{i}=h_{\varepsilon(i)}$ for all $i\in B$.
}
\end{theorem}

\setlength{\parindent}{2em}
Theorem 2.5 has the following straightforward corollary, which establishes Conjecture 1.2 with the additional assumption that $\mathbf{P}$ is hierarchical.

\begin{corollary}
{If $\mathbf{P}$ is hierarchical, then $\Lambda$ is finer than $\mathcal{Q}(\hat{\mathbf{H}},\mathbf{\overline{P}})$.
}
\end{corollary}

\begin{proof}
Let $\alpha,\gamma\in\hat{\mathbf{H}}$ such that $\alpha\sim_{\Lambda}\gamma$. By Theorem 2.5, we deduce that $|\langle \supp(\alpha)\rangle_{\mathbf{\overline{P}}}|=|\langle \supp(\gamma)\rangle_{\mathbf{\overline{P}}}|$, and hence $\wt_{\mathbf{\overline{P}}}(\alpha)=\wt_{\mathbf{\overline{P}}}(\gamma)$. It immediately follows that $\Lambda$ is finer than $\mathcal{Q}(\hat{\mathbf{H}},\mathbf{\overline{P}})$, completing the proof.
\end{proof}

\setlength{\parindent}{2em}
We note that neither Theorem 2.5 nor Corollary 2.2 requires the assumption that $h_{i}\geqslant2$ for all $i\in\Omega$. We will establish a generalization of Theorem 2.5 in Section 3.2, and then prove Theorem 2.5 in Section 3.4.

\section{Generalizing $F(\alpha)$}
\setlength{\parindent}{2em}
In this section, we study a generalized version of $F(\alpha)$ defined in Section 2. Throughout this section, we let $\Omega$ be a finite set with $|\Omega|=n$ and let $\mathbf{P}=(\Omega,\preccurlyeq_{\mathbf{P}})$ be a poset.

We begin with the following notation.

\begin{notation}
Let $Y\subseteq \Omega$. Then, $\max_{\mathbf{P}}(Y)$ and $\min_{\mathbf{P}}(Y)$ will be denoted by $\max(Y)$ and $\min(Y)$, respectively. For the poset $\mathbf{Q}=(Y,\preccurlyeq_{\mathbf{Q}})$ defined as
$$\mbox{$u\preccurlyeq_{\mathbf{Q}} v\Longleftrightarrow u\preccurlyeq_{\mathbf{P}}v$ for all $u,v\in Y$},$$ 
$\mathcal{I}(\mathbf{Q})$ and $\mathcal{I}(\mathbf{\overline{Q}})$ will be denoted by $\mathcal{I}(Y)$ and $\mathcal{I}^{\mathbf{c}}(Y)$, respectively. Furthermore, for any $A\subseteq Y$, $\langle A\rangle_{\mathbf{Q}}$ will be denoted by $\langle A\rangle_{Y}$.
\end{notation}

Now let $K$ be a commutative ring and fix $\tau,\eta\in K^{\Omega}$. For any $D,I\subseteq\Omega$, we define $\varphi(D,I)\in K$ as
\begin{equation}
\hspace*{-2mm}\begin{cases}
(-1)^{|I\cap D|}\left(\prod_{i\in I-\max(I)}\tau_{(i)}\right)\left(\prod_{i\in \max(I)-D}\eta_{(i)}\right),&I\cap D\subseteq \max(I);\\
0,&I\cap D\not\subseteq \max(I).
\end{cases}
\end{equation}
For any $Y\subseteq\Omega$ and $D\subseteq Y$, we define $\pi(Y,D)\in K[x]$ as follows,
\begin{equation}\text{$\pi(Y,D)=\sum_{I\in\mathcal{I}(Y)}\varphi(D,I) x^{|I|}$.}\end{equation}
Finally, for $f\in K[x]$ and $i\in\mathbb{N}$, we let $f_{(i)}$ denote the coefficient of $x^{i}$ in $f$. Then, for $Y\subseteq\Omega$, $D\subseteq Y$ and $l\in\mathbb{N}$, it is straightforward to verify that
\begin{equation}\text{$(\pi(Y,D))_{(l)}=\sum_{(I\in\mathcal{I}(Y),|I|=l)}\varphi(D,I)$.}\end{equation}

\subsection{Some basic properties of $\pi(\cdot,\cdot)$}

\setlength{\parindent}{2em}
We begin by establishing the following analogue of Theorem 2.1.

\setlength{\parindent}{0em}
\begin{theorem}
For $X\subseteq\Omega$ and $D\in\mathcal{I}^{\mathbf{c}}(X)$, let $Y=\min(D)\cup(X-D)$. Then, we have $\min(D)\subseteq\max(Y)$, and it holds that
\begin{eqnarray*}
\begin{split}
\pi(X,D)&=\pi(Y,\min(D))\\
&=\sum_{I\in\mathcal{I}(Y)}(-1)^{|I\cap \min(D)|}\left(\prod_{i\in I-\max(I)}\tau_{(i)}\right)\left(\prod_{i\in \max(I)-\min(D)}\eta_{(i)}\right)x^{|I|}.
\end{split}
\end{eqnarray*}
\end{theorem}

\begin{proof}
First of all, from $D\in\mathcal{I}^{\mathbf{c}}(X)$, we deduce that $\min(D)\subseteq\max(Y)$. Now applying Lemma 2.3 to the poset $\mathbf{Q}=(X,\preccurlyeq_{\mathbf{Q}})$ defined as
$$\mbox{$u\preccurlyeq_{\mathbf{Q}} v\Longleftrightarrow u\preccurlyeq_{\mathbf{P}}v$ for all $u,v\in X$},$$ 
for any $I\in\mathcal{I}(X)$, we deduce that
$$\text{$I\cap D\subseteq \max(I)\Longleftrightarrow I\subseteq Y$}.$$
Since $Y\in \mathcal{I}(X)$, we deduce that $\mathcal{I}(X)\cap 2^{Y}=\mathcal{I}(Y)$. Moreover, for any $I\in\mathcal{I}(Y)$, it can be readily verified that $\varphi(D,I)=\varphi(\min(D),I)$. Now the result immediately follows from (3.1) and (3.2).
\end{proof}

\begin{remark}
Adopting the notations in Section 2, if we set $K=\mathbb{R}$ and $\tau_{(i)}=h_{i}$, $\eta_{(i)}=h_{i}-1$ for all $i\in\Omega$, then for $\alpha\in \hat{\mathbf{H}}$ with $\langle \supp(\alpha)\rangle_{\mathbf{\overline{P}}}=D$, by Theorem 2.1 and Theorem 3.1, we have $\pi(\Omega,D)=F(\alpha)$. Therefore, the notion $\pi(Y,D)$ can be regarded as a generalization of the notion $F(\alpha)$.
\end{remark}

\setlength{\parindent}{2em}
Parallel to Proposition 2.2, we have the following corollary of Theorem 3.1 concerning the degree and leading coefficient of $\pi(\cdot,\cdot)$.

\begin{corollary}
If $K$ is a field and $\tau_{(i)}\neq0$, $\eta_{(i)}\neq0$ for all $i\in\Omega$, then for $X\subseteq\Omega$ and $D\in\mathcal{I}^{\mathbf{c}}(X)$, we have $\deg(\pi(X,D))=|X|-|D|+|\min(D)|$, and the leading coefficient of $\pi(X,D)$ is $\varphi(\min(D),(X-D)\cup\min(D))$.
\end{corollary}

\setlength{\parindent}{2em}
The following Lemma will be used frequently in our proofs.

\setlength{\parindent}{0em}
\begin{lemma}
For $Y\subseteq\Omega$, $D\in\mathcal{I}^{\mathbf{c}}(Y)$ and $e\in \max(Y)$, it holds that
\begin{eqnarray*}
\begin{split}
&\pi(Y,D)-\pi(Y-\{e\},D-\{e\})\\
&=\mbox{$\begin{cases}
\eta_{(e)}\left(\prod_{i\in \langle\{e\}\rangle_{Y}-\{e\}}\tau_{(i)}\right) x^{|\langle\{e\}\rangle_{Y}|}\pi(Y-\langle\{e\}\rangle_{Y},D),&e\not\in D;\\
0,&e\in D-\min(D);\\
-\left(\prod_{i\in \langle\{e\}\rangle_{Y}-\{e\}}\tau_{(i)}\right) x^{|\langle\{e\}\rangle_{Y}|}\pi(Y-\langle\{e\}\rangle_{Y},D-\{e\}),&e\in \min(D).
\end{cases}$}
\end{split}
\end{eqnarray*}
\end{lemma}

\begin{proof}
Throughout the proof, we let $U=\langle\{e\}\rangle_{Y}$. First, if $e\in D-\min(D)$, then we have $\min(D-\{e\})=\min(D)$ and $(Y-\{e\})-(D-\{e\})=Y-D$, which, together with Theorem 3.1, implies that $\pi(Y-\{e\},D-\{e\})=\pi(Y,D)$. Therefore in the following, we assume either $e\not\in D$ or $e\in \min(D)$. Then, it is straightforward to verify that
\begin{equation}\text{$\left(e\in \min(D)\Longrightarrow D\cap U=\{e\}\right)$ and $\left(e\not\in D\Longrightarrow D\cap U=\emptyset\right)$}.\end{equation}
By (3.2), we have
$$\pi(Y,D)=\left(\sum_{(I\in\mathcal{I}(Y),e\in I)}\varphi(D,I) x^{|I|}\right)+\left(\sum_{(I\in\mathcal{I}(Y),e\not\in I)}\varphi(D,I) x^{|I|}\right).$$
Since $e\in \max(Y)$, we have $\{I\mid I\in\mathcal{I}(Y),~e\not\in I\}=\mathcal{I}(Y-\{e\})$. For any $I\in\mathcal{I}(Y-\{e\})$, by straightforward computation, we deduce that $\varphi(D,I)=\varphi(D-\{e\},I)$. It then follows that
$$\sum_{(I\in\mathcal{I}(Y),e\not\in I)}\varphi(D,I)x^{|I|}=\sum_{I\in\mathcal{I}(Y-\{e\})}\varphi(D-\{e\},I)x^{|I|}=\pi(Y-\{e\},D-\{e\}).$$
Also noticing that for any $I\in\mathcal{I}(Y)$, $e\in I\Longleftrightarrow U\subseteq I$, we deduce that
\begin{equation}\pi(Y,D)-\pi(Y-\{e\},D-\{e\})=\sum_{(I\in\mathcal{I}(Y),U\subseteq I)}\varphi(D,I) x^{|I|}.\end{equation}
For any $I\in\mathcal{I}(Y)$ such that $U\subseteq I$, by (3.4) and some straightforward computation, we deduce that
$$\varphi(D,I)=(-1)^{|D\cap U|}\left(\prod_{i\in U-\{e\}}\tau_{(i)}\right)\left(\prod_{i\in \{e\}-D}\eta_{(i)}\right)\varphi(D-U,I-U).$$
The above discussion, together with $U\in\mathcal{I}(Y)$, implies that
\begin{eqnarray*}
\begin{split}
&\sum_{(I\in\mathcal{I}(Y),U\subseteq I)}\varphi(D,I) x^{|I|}\\
&=(-1)^{|D\cap U|}\left(\prod_{i\in U-\{e\}}\tau_{(i)}\right)\left(\prod_{i\in \{e\}-D}\eta_{(i)}\right)x^{|U|}\left(\sum_{(I\in\mathcal{I}(Y),U\subseteq I)}\varphi(D-U,I-U) x^{|I-U|}\right)\\
&=(-1)^{|D\cap U|}\left(\prod_{i\in U-\{e\}}\tau_{(i)}\right)\left(\prod_{i\in \{e\}-D}\eta_{(i)}\right)x^{|U|}\left(\sum_{J\in\mathcal{I}(Y-U)}\varphi(D-U,J) x^{|J|}\right)\\
&=(-1)^{|D\cap U|}\left(\prod_{i\in U-\{e\}}\tau_{(i)}\right)\left(\prod_{i\in \{e\}-D}\eta_{(i)}\right)x^{|U|}\pi(Y-U,D-U).
\end{split}
\end{eqnarray*}
Now the lemma immediately follows from (3.4) and (3.5).
\end{proof}

\setlength{\parindent}{2em}
Lemma 3.1 immediately implies the following proposition.

\begin{proposition}
For $Y\subseteq\Omega$, $L\in\mathcal{I}^{\mathbf{c}}(Y)$ and $e\in \max(Y)-L$, let $k=|\langle\{e\}\rangle_{Y}|$. Then, we have
\begin{eqnarray*}
\begin{split}
\pi(Y,L)-\pi(Y,L\cup\{e\})=\mbox{$(\eta_{(e)}+1_{K})\left(\prod_{i\in \langle\{e\}\rangle_{Y}-\{e\}}\tau_{(i)}\right)x^{k}\pi(Y-\langle\{e\}\rangle_{Y},L)$.}
\end{split}
\end{eqnarray*}
\end{proposition}

\begin{proof}
Since $L\in\mathcal{I}^{\mathbf{c}}(Y)$, $e\in \max(Y)-L$, we deduce that $L\cup\{e\}\in\mathcal{I}^{\mathbf{c}}(Y)$, $e\in\min(L\cup\{e\})$. Now the result follows from applying Lemma 3.2 to $(Y,L)$ and $(Y,L\cup\{e\})$ respectively.
\end{proof}

\setlength{\parindent}{2em}
Now we consider the special case that
\begin{equation}\text{$\eta_{(i)}=\tau_{(i)}-1_{K}$ for all $i\in \Omega$}.\end{equation}
We note that by Remark 3.1, (3.6) is satisfied by $F(\alpha)$ defined in Section 2.

\setlength{\parindent}{0em}
\begin{theorem}
Assume (3.6) holds, and let $Y\subseteq\Omega$. Then, it holds that

{\bf{(1)}}\,\,For any $D\subseteq\max(Y)$,
$$\pi(Y,D)=\sum_{A\subseteq D}(-1)^{|A|}\left(\prod_{i\in \langle A\rangle_{Y}}\tau_{(i)}x\right)\pi(Y-\langle A\rangle_{Y},\emptyset);$$
{\bf{(2)}}\,\,For any $D\subseteq\max(Y)$,
$$\left(\prod_{i\in \langle D\rangle_{Y}}\tau_{(i)}x\right) \pi(Y-\langle D\rangle_{Y},\emptyset)=\sum_{A\subseteq D}(-1)^{|A|}\pi(Y,A).$$
\end{theorem}

\begin{proof}
(1)\,\,Let $D\subseteq\max(Y)$. If $D=\emptyset$, then the result trivially holds. Therefore in the following, we assume $D\neq\emptyset$. Fix $e\in D$ and let $L=D-\{e\}$. Then, by Proposition 3.1, we deduce that
\begin{equation}\pi(Y,D)=\pi(Y,L)-\left(\prod_{i\in \langle\{e\}\rangle_{Y}}\tau_{(i)}x\right) \pi(Y-\langle\{e\}\rangle_{Y},L).\end{equation}
Applying an induction argument to $(Y,L)$ and $(Y-\langle\{e\}\rangle_{Y},L)$ respectively, (1) follows from (3.7) and some straightforward computation which we omit.

(2)\,\,This immediately follows from (1) and the principle of inclusion-exclusion (see [15, Theorem 2.1.1]).
\end{proof}

\setlength{\parindent}{2em}
A combination of Theorems 3.1 and 3.2 yields the following result.

\setlength{\parindent}{0em}
\begin{corollary}
Assume (3.6) holds, and let $X\subseteq\Omega$, $D\in\mathcal{I}^{\mathbf{c}}(X)$. Then, it holds that
$$\pi(X,D)=\sum_{A\subseteq \min(D)}(-1)^{|A|}\left(\prod_{i\in \langle A\rangle_{X}}\tau_{(i)}x\right)\pi((X-\langle A\rangle_{X})-(D-\min(D)),\emptyset).$$
\end{corollary}

\begin{proof}
Let $Y=(X-D)\cup\min(D)$. By Theorem 3.1, we have $\pi(X,D)=\pi(Y,\min(D))$ and $\min(D)\subseteq\max(Y)$, which, together with (1) of Theorem 3.2, implies that
$$\pi(X,D)=\pi(Y,\min(D))=\sum_{A\subseteq \min(D)}(-1)^{|A|}\left(\prod_{i\in \langle A\rangle_{Y}}\tau_{(i)}x\right)\pi(Y-\langle A\rangle_{Y},\emptyset).$$
Now, for any $A\subseteq \min(D)$, it is straightforward to verify that $\langle A\rangle_{Y}=\langle A\rangle_{X}$, $Y-\langle A\rangle_{Y}=(X-\langle A\rangle_{X})-(D-\min(D))$, which immediately implies the corollary.
\end{proof}

\subsection{Conditions for $\mathbf{P}$ to be hierarchical}

\setlength{\parindent}{2em}
In this subsection, we present some necessary or sufficient conditions for $\mathbf{P}$ to be hierarchical.

\setlength{\parindent}{2em}
We begin with the following generalization of Theorem 2.3.

\setlength{\parindent}{0em}
\begin{proposition}
If $K$ is a field and $\tau_{(i)}\neq0$, $\eta_{(i)}\neq0$ for all $i\in\Omega$, then the following two conditions are equivalent to each other:

{\bf{(1)}}\,\,$\mathbf{P}$ is hierarchical;

{\bf{(2)}}\,\,For any $B,D\in\mathcal{I}(\mathbf{\overline{P}})$, $|B|=|D|\Longrightarrow\deg(\pi(\Omega,B))=\deg(\pi(\Omega,D))$.
\end{proposition}

\begin{proof}
For any $B,D\in\mathcal{I}(\mathbf{\overline{P}})$ with $|B|=|D|$, by Corollary 3.1, we have
$$\deg(\pi(\Omega,D))=\deg(\pi(\Omega,B))\Longleftrightarrow|\min(D)|=|\min(B)|.$$
Now the result immediately follows from Lemma 2.6.
\end{proof}

\begin{remark}
By Remark 3.1, it is straightforward to verify that Proposition 3.2 generalizes Theorem 2.3.
\end{remark}

\setlength{\parindent}{2em}
Now we establish a sufficient condition for $\mathbf{P}$ to be hierarchical with the assumption that $K$ is a field and $\tau_{(i)}\neq0$ for all $i\in\Omega$. The following proposition is an analogue of [7, Lemma 2.3] and [13. Lemma 1].

\setlength{\parindent}{0em}
\begin{proposition}
Suppose that $K$ is a field, $\tau_{(i)}\neq0$ for all $i\in\Omega$ and that
$$\mbox{$|B|=|D|\Longrightarrow\pi(\Omega,B)=\pi(\Omega,D)$ for all $B,D\in\mathcal{I}(\mathbf{\overline{P}})$}.$$
Then, $\mathbf{P}$ is hierarchical.
\end{proposition}

\begin{proof}
Fixing $A\subseteq\Omega$ such that $a\preccurlyeq_{\mathbf{P}}b$ for all $a\in A$, $b\in \Omega-A$, we will show
\begin{equation}\text{$c\preccurlyeq_{\mathbf{P}}d$ for all $c\in A-\max(A)$, $d\in \max(A)$}.\end{equation}
First of all, it is straightforward to verify that $A\in\mathcal{I}(\mathbf{P})$, and
\begin{equation}\forall~I\in\mathcal{I}(\mathbf{P}),~|I|=|A|\Longrightarrow I=A.\end{equation}
Let $B=(\Omega-A)\cup\max(A)$. Then, it can be readily verified that $B\in\mathcal{I}(\mathbf{\overline{P}})$, $B\cap A=\max(A)$, which, together with (3.9) and the fact $\tau_{(i)}\neq0$ for all $i\in\Omega$, yields that $\left(\pi(\Omega,B)\right)_{(|A|)}=\varphi(B,A)\neq0$. Now, by way of contradiction, we assume that (3.8) does not hold. Then, we can choose $u\in\max(A-\max(A))$ and $y\in \max(A)$ such that $u\not\preccurlyeq_{\mathbf{P}}y$. Let $D=\left(B-\{y\}\right)\cup\{u\}$. Then, we can verify that $D\in\mathcal{I}(\mathbf{\overline{P}})$, $|D|=|B|$. By $u\in D\cap A$, $u\not\in \max(A)$ and (3.9), we deduce that $\left(\pi(\Omega,D)\right)_{(|A|)}=\varphi(D,A)=0$. Hence, we conclude that $\pi(\Omega,B)\neq\pi(\Omega,D)$, $B,D\in\mathcal{I}(\mathbf{\overline{P}})$, $|B|=|D|$, a contradiction, as desired. It then follows that (3.8) holds true. Now we apply Lemma 2.6 and reach $\mathbf{P}$ is hierarchical, completing the proof.
\end{proof}

\setlength{\parindent}{2em}
For further discussion, throughout the rest of this subsection, we let $m$ denote the largest cardinality of a chain in $\mathbf{P}$, and for any $j\in[1,m]$, we let $W_{j}=\{u\mid u\in\Omega,~\len(u)=j\}$. Moreover, for any $D\subseteq\Omega$, we let $\sigma(D)$ denote the integer $r\in[1,m]$ such that $D\subseteq \bigcup_{j=r}^{m}W_{j}$ and for any $p\in[r+1,m]$, $D\not\subseteq \bigcup_{j=p}^{m}W_{j}$.

\setlength{\parindent}{0em}
\begin{lemma}
If $\mathbf{P}$ is hierarchical, then for $D\in\mathcal{I}(\mathbf{\overline{P}})$ with $\sigma(D)=r$, $\pi(\Omega,D)$ is equal to
\begin{eqnarray*}
\begin{split}
&\left(\prod_{i\in \left(\bigcup_{j=1}^{r-1}W_{j}\right)}\tau_{(i)}x\right)\left((1_{K}-x)^{|W_{r}\cap D|}\left(\prod_{i\in W_{r}-D}\left(1_{K}+\eta_{(i)}x\right)\right)-1_{K}\right)\\
&+\left(\sum_{t=1}^{r-1}\left(\prod_{i\in \left(\bigcup_{j=1}^{t-1}W_{j}\right)}\tau_{(i)}x\right)\left(\left(\prod_{i\in W_{t}}\left(1_{K}+\eta_{(i)} x\right)\right)-1_{K}\right)\right)+1_{K}.
\end{split}
\end{eqnarray*}
\end{lemma}

\begin{proof}
We define $g:\bigcup_{t=1}^{m}\{t\}\times \left(2^{W_{t}}-\{\emptyset\}\right)\longrightarrow 2^{\Omega}$ as $$\mbox{$g(t,V)=\left(\bigcup_{j=1}^{t-1}W_{j}\right)\cup V$ for all $t\in[1,m]$, $V\subseteq W_{t}$, $V\neq\emptyset$}.$$
Since $\mathbf{P}$ is hierarchical, by Lemma 2.5 and (3) of Lemma 2.2, we deduce that $g$ is injective and the range of $g$ is equal to $\mathcal{I}(\mathbf{P})-\{\emptyset\}$. Again by Lemma 2.5, we deduce that $D=(D\cap W_{r})\cup(\bigcup_{j=r+1}^{m}W_{j})$ and $D\cap W_{r}=\min_{\mathbf{P}}(D)$. Now the result follows from Theorem 3.1, together with some straightforward computation which we omit.
\end{proof}

\setlength{\parindent}{2em}
The following proposition generalizes Theorem 2.5.

\setlength{\parindent}{0em}
\begin{proposition}
Assume that $\mathbf{P}$ is hierarchical, $K$ is a field and $\tau_{(i)}\neq0$, $\eta_{(i)}\neq-1_{K}$ for all $i\in \Omega$. Fix $D,C\in\mathcal{I}(\mathbf{\overline{P}})$. Then, the following two conditions are equivalent to each other:

{\bf{(1)}}\,\,$\pi(\Omega,C)=\pi(\Omega,D)$;

{\bf{(2)}}\,\,There exists a bijection $\varepsilon:C\longrightarrow D$ such that $\eta_{(i)}=\eta_{(\varepsilon(i))}$ for all $i\in C$.
\end{proposition}

\begin{proof}
Throughout the proof, we let $r=\sigma(D)$.

$(1)\Longrightarrow(2)$\,\,First, we show $\sigma(C)=r$. Without loss of generality, we assume $r\leqslant \sigma(C)$. If $r+1\leqslant \sigma(C)$, then by Lemma 3.2 and the fact that $K$ is a field and $\tau_{(i)}\neq0$ for all $i\in \Omega$, we deduce that
$$(1_{K}-x)^{|W_{r}\cap D|}\left(\prod_{i\in W_{r}-D}\left(1_{K}+\eta_{(i)}x\right)\right)=\prod_{i\in W_{r}}\left(1_{K}+\eta_{(i)}x\right).$$
Since $\eta_{(i)}\neq-1_{K}$ for all $i\in \Omega$, we deduce that $|W_{r}\cap D|=0$ and hence $W_{r}\cap D=\emptyset$. Noticing that $r=\sigma(D)$, we deduce that $r=m$, contrary to $r+1\leqslant \sigma(C)$. The above discussion implies that $\sigma(C)=r$. Now again by Lemma 3.2, we deduce that
$$\hspace*{-4mm}(1_{K}-x)^{|W_{r}\cap D|}\left(\prod_{i\in W_{r}-D}\left(1_{K}+\eta_{(i)}x\right)\right)=(1_{K}-x)^{|W_{r}\cap C|}\left(\prod_{i\in W_{r}-C}\left(1_{K}+\eta_{(i)}x\right)\right).$$
Noticing that $\eta_{(i)}\neq-1_{K}$ for all $i\in \Omega$, we deduce that $|W_{r}\cap C|=|W_{r}\cap D|$, $\prod_{i\in W_{r}-C}\left(1_{K}+\eta_{(i)} x\right)=\prod_{i\in W_{r}-D}\left(1_{K}+\eta_{(i)} x\right)$, which further implies that
$$\prod_{i\in W_{r}\cap C}\left(1_{K}+\eta_{(i)} x\right)=\prod_{i\in W_{r}\cap D}\left(1_{K}+\eta_{(i)} x\right).$$
Hence, there exists a bijection $\mu:W_{r}\cap C\longrightarrow W_{r}\cap D$ such that $\eta_{(i)}=\eta_{(\mu(i))}$ for all $i\in W_{r}\cap C$. Now define $\varepsilon=\mu\cup\{(i,i)\mid i\in\bigcup_{j=r+1}^{m}W_{j}\}$. Since $\mathbf{P}$ is hierarchical, by Lemma 2.5, we deduce that $\varepsilon$ is a bijection from $C$ to $D$ such that $\eta_{(i)}=\eta_{(\varepsilon(i))}$ for all $i\in C$, proving (2).

$(2)\Longrightarrow(1)$\,\,The existence of $\varepsilon$ implies that $|C|=|D|$ and
\begin{equation}\prod_{i\in C}\left(1_{K}+\eta_{(i)} x\right)=\prod_{i\in D}\left(1_{K}+\eta_{(i)} x\right).\end{equation}
Since $\mathbf{P}$ is hierarchical, by Lemma 2.5, we deduce that $\sigma(C)=\sigma(D)=r$, $|W_{r}\cap C|=|W_{r}\cap D|$. Moreover, by Lemma 2.5 and (3.10), we deduce that
$$\prod_{i\in W_{r}\cap C}\left(1_{K}+\eta_{(i)} x\right)=\prod_{i\in W_{r}\cap D}\left(1_{K}+\eta_{(i)} x\right),$$
which further implies that
$$\prod_{i\in W_{r}-C}\left(1_{K}+\eta_{(i)} x\right)=\prod_{i\in W_{r}-D}\left(1_{K}+\eta_{(i)} x\right).$$
Now (1) immediately follows from Lemma 3.2.
\end{proof}

\setlength{\parindent}{2em}
Proposition 3.4 will be used to prove Theorem 2.5 in Section 3.4. Now, combining Propositions 3.3 and 3.4, we have the following theorem.

\setlength{\parindent}{0em}
\begin{theorem}
Assume that $K$ is a field and $\tau_{(i)}\neq0$, $\eta_{(i)}\neq-1_{K}$ for all $i\in \Omega$. Then, the following three conditions are equivalent to each other:

{\bf{(1)}}\,\,For any $C,D\in\mathcal{I}(\mathbf{\overline{P}})$, $|C|=|D|\Longrightarrow\pi(\Omega,C)=\pi(\Omega,D)$;

{\bf{(2)}}\,\,For any $C,D\in\mathcal{I}(\mathbf{\overline{P}})$, $\pi(\Omega,C)=\pi(\Omega,D)\Longleftrightarrow |C|=|D|$;

{\bf{(3)}}\,\,$\mathbf{P}$ is hierarchical, and for any $u,v\in\Omega$ with $\len(u)=\len(v)$, it holds that $\eta_{(u)}=\eta_{(v)}$.
\end{theorem}

\begin{proof}
$(2)\Longrightarrow(1)$\,\,This is trivial.

$(3)\Longrightarrow(1)$\,\,We fix $C,D\in\mathcal{I}(\overline{\mathbf{P}})$ such that $|C|=|D|$. Since $\mathbf{P}$ is hierarchical, by Lemma 2.5, we deduce that $\sigma(D)=\sigma(C)=r$ for some $r\in[1,m]$, and moreover, it holds that $|W_{r}\cap C|=|W_{r}\cap D|$. Now let $\mu:W_{r}\cap C\longrightarrow W_{r}\cap D$ be any fixed bijective map, and define $\varepsilon=\mu\cup\{(i,i)\mid i\in\bigcup_{j=r+1}^{m}W_{j}\}$. Since $\mathbf{P}$ is hierarchical, by Lemma 2.5, together with the assumption in (3), we deduce that $\varepsilon$ is a bijection from $C$ to $D$ such that $\eta_{(i)}=\eta_{(\varepsilon(i))}$ for all $i\in C$. Now by Proposition 3.4, we conclude that $\pi(\Omega,C)=\pi(\Omega,D)$, proving (1).

$(1)\Longrightarrow(2)$ $\wedge$ $(3)$\,\,First, by Proposition 3.3, we deduce that $\mathbf{P}$ is hierarchical. Now (2) immediately follows from Proposition 3.4. Next, for a fixed $r\in[1,m]$ and $b,c\in W_{r}$, we show that $\eta_{(b)}=\eta_{(c)}$. Define $D,C\in\mathcal{I}(\mathbf{\overline{P}})$ as
$$\mbox{$D=\{b\}\cup(\bigcup_{j=r+1}^{m}W_{j})$, $C=\{c\}\cup(\bigcup_{j=r+1}^{m}W_{j})$.}$$
Clearly, we have $|C|=|D|$. By (1), we deduce that $\pi(X,C)=\pi(X,D)$, which, together with Proposition 3.4, implies that
$$|\{i\mid i\in D,\eta_{(i)}=\eta_{(c)}\}|=|\{i\mid i\in C,\eta_{(i)}=\eta_{(c)}\}|,$$
which further implies that
$$|\{i\mid i\in \{b\},\eta_{(i)}=\eta_{(c)}\}|=|\{i\mid i\in \{c\},\eta_{(i)}=\eta_{(c)}\}|=1.$$
Therefore we conclude that $\eta_{(b)}=\eta_{(c)}$, proving (3).
\end{proof}

\begin{remark}
With the help of Lemma 2.1 and Remark 3.1, Theorem 3.3 recovers [4, Theorems 5.4 and 5.5].
\end{remark}

\subsection{The case $K=\mathbb{R}$}

\setlength{\parindent}{2em}
Throughout this subsection, we focus on the special case that $K=\mathbb{R}$, where we further assume that
\begin{equation}\text{$\tau_{(i)}>0$, $\eta_{(i)}>-1$, $\eta_{(i)}\neq0$ for all $i\in\Omega$}.\end{equation}

We begin by defining a total order on $\mathbb{R}[x]$, denoted by $\curlyeqprec$, such that for any $f,g\in\mathbb{R}[x]$, $f\curlyeqprec g$ if and only if one of the following three conditions holds:

{\bf{1)}}\,\,$f=g$;

{\bf{2)}}\,\,$\deg(f)\leqslant\deg(g)-1$;

{\bf{3)}}\,\,$\deg(f)=\deg(g)$, and there exists $k\in\mathbb{N}$ such that $f_{(k)}<g_{(k)}$ and
$$f\equiv g~(\bmod~x^{k}).$$

\setlength{\parindent}{0em}
\begin{lemma}
For $Y\subseteq\Omega$, $A\subsetneqq \max(Y)$ and $e\in\max(Y)-A$, we have
$$\pi(Y,A\cup\{e\})\neq\pi(Y,A),~\pi(Y,A\cup\{e\})\curlyeqprec\pi(Y,A).$$
\end{lemma}

\begin{proof}
By Corollary 3.1, we have $\deg(\pi(Y,A\cup\{e\}))=|Y|=\deg(\pi(Y,A))$. Now let $k=|\langle\{e\}\rangle_{Y}|$. Then, by Proposition 3.1 and (3.11), we deduce that $\pi(Y,A)\equiv \pi(Y,A\cup\{e\})~(\bmod~x^{k})$, $(\pi(Y,A))_{(k)}>(\pi(Y,A\cup\{e\}))_{(k)}$. Now the lemma immediately follows from the definition of $\curlyeqprec$.
\end{proof}

\setlength{\parindent}{2em}
Now we are ready to prove the first main result in this subsection.

\setlength{\parindent}{0em}
\begin{theorem}
For $X\subseteq\Omega$ and $D,A\in\mathcal{I}^{\mathbf{c}}(X)$ with $A\subsetneqq D$, we have
$$\pi(X,D)\neq\pi(X,A),~\pi(X,D)\curlyeqprec\pi(X,A).$$
\end{theorem}

\begin{proof}
First, by $\min(D)\subseteq(D-A)\cup\min(A)$, we have
\begin{equation}|A|-|\min(A)|\leqslant|D|-|\min(D)|,\end{equation}
and moreover,
\begin{equation}|A|-|\min(A)|=|D|-|\min(D)|\Longleftrightarrow\min(D)=(D-A)\cup\min(A).\end{equation}
By Corollary 3.1 and (3.12), we have $\deg(\pi(X,D))\leqslant\deg(\pi(X,A))$. Hence, by the definition of $\curlyeqprec$, we can assume $\deg(\pi(X,D))=\deg(\pi(X,A))$. Then, by Corollary 3.1, we have $|A|-|\min(A)|=|D|-|\min(D)|$, which, together with (3.13), implies that $\min(D)=(D-A)\cup\min(A)$, and furthermore,
$$(X-D)\cup\min(D)=(X-A)\cup\min(A):=Y.$$
By Theorem 3.1, we have
$$\pi(X,D)=\pi(Y,\min(D)),~\pi(X,A)=\pi(Y,\min(A)),~\min(D)\subseteq\max(Y).$$
Also noticing that $A\subsetneqq D$, $\min(D)=(D-A)\cup\min(A)$, we deduce that $\min(A)\subsetneqq\min(D)$. Now the result immediately follows from Lemma 3.3.
\end{proof}

\setlength{\parindent}{2em}
The following theorem is our second main result of this subsection, which, as detailed in Section 3.4, can be used to prove Theorem 2.4.

\setlength{\parindent}{0em}
\begin{theorem}
Let $\Theta=\{\pi(\Omega,M)\mid M\in\mathcal{I}(\mathbf{\overline{P}})\}$. Then, we have $|\Theta|\geqslant n+1$. Moreover, the following four conditions are equivalent to each other:

{\bf{(1)}}\,\,$|\Theta|=n+1$;

{\bf{(2)}}\,\,For any $C,D\in\mathcal{I}(\mathbf{\overline{P}})$, $|C|=|D|\Longrightarrow\pi(\Omega,C)=\pi(\Omega,D)$;

{\bf{(3)}}\,\,For any $C,D\in\mathcal{I}(\mathbf{\overline{P}})$, $\pi(\Omega,C)=\pi(\Omega,D)\Longleftrightarrow |C|=|D|$;

{\bf{(4)}}\,\,$\mathbf{P}$ is hierarchical, and for any $u,v\in \Omega$ with $\len(u)=\len(v)$, it holds that $\eta_{(u)}=\eta_{(v)}$.
\end{theorem}

\begin{proof}
By Lemma 2.2 and Theorem 3.4, it can be readily verified that $|\Theta|\geqslant|\Omega|+1=n+1$. Now, $(2)\Longleftrightarrow(3)\Longleftrightarrow(4)$ follows from Theorem 3.3, and Lemma 2.2 immediately yields that $(3)\Longrightarrow(1)$. Therefore it remains to prove $(1)\Longrightarrow(2)$.

$(1)\Longrightarrow(2)$\,\,We fix $C,D\in\mathcal{I}(\mathbf{\overline{P}})$ with $|C|=|D|=m$. By Lemma 2.2, we can choose a sequence $A_{(0)}\subseteq A_{(1)}\subseteq\cdots \subseteq  A_{(n)}$ such that $A_{(m)}=C$ and for any $j\in[0,n]$, $A_{(j)}\in\mathcal{I}(\mathbf{\overline{P}})$, $|A_{(j)}|=j$. Similarly, we can choose a sequence $B_{(0)}\subseteq B_{(1)}\subseteq\cdots \subseteq B_{(n)}$ such that $B_{(m)}=D$ and for any $j\in[0,n]$, $B_{(j)}\in\mathcal{I}(\mathbf{\overline{P}})$, $|B_{(j)}|=j$. Now define $\alpha:[0,n]\longrightarrow\Theta$ as $\alpha(j)=\pi(\Omega,A_{(j)})$, and define $\beta:[0,n]\longrightarrow\Theta$ as $\beta(j)=\pi(\Omega,B_{(j)})$. By Theorem 3.4, for any $j,t\in[0,n]$ with $j+1\leqslant t$, we have $\alpha(t)\neq\alpha(j)$, $\alpha(t)\curlyeqprec\alpha(j)$, $\beta(t)\neq\beta(j),~\beta(t)\curlyeqprec\beta(j)$. Noticing that $|[0,n]|=|\Theta|=n+1$ and $\curlyeqprec$ is a total order on $\mathbb{R}[x]$, we deduce that $\alpha=\beta$. In particular, we have
$$\pi(\Omega,C)=\pi(\Omega,A_{(m)})=\alpha(m)=\beta(m)=\pi(\Omega,B_{(m)})=\pi(\Omega,D),$$
which immediately establishes (2).
\end{proof}

\subsection{Proofs of Theorem 2.4 and Theorem 2.5}

\setlength{\parindent}{2em}
In this subsection, we prove Theorem 2.4 and Theorem 2.5. Throughout, we assume $K=\mathbb{R}$, and we will use the notations in Section 2. As in Remark 3.1, we set $\tau_{(i)}=h_{i}$, $\eta_{(i)}=h_{i}-1$ for all $i\in\Omega$, and moreover, we define $\Theta$ as
$$\Theta=\{\pi(\Omega,M)\mid M\in\mathcal{I}(\mathbf{\overline{P}})\}.$$

We begin with the following lemma.

\setlength{\parindent}{0em}
\begin{lemma}
{\bf{(1)}}\,\,For $\alpha,\gamma\in\hat{\mathbf{H}}$ with $D=\langle \supp(\alpha)\rangle_{\mathbf{\overline{P}}}$, $B=\langle \supp(\gamma)\rangle_{\mathbf{\overline{P}}}$, it holds that $\alpha\sim_{\Lambda}\gamma\Longleftrightarrow\pi(\Omega,D)=\pi(\Omega,B)$.

{\bf{(2)}}\,\,If $h_{i}\geqslant2$ for all $i\in\Omega$, then $|\Lambda|=|\Theta|$, and moreover, $\mathcal{Q}(\mathbf{H},\mathbf{P})$ is Fourier-reflexive if and only if $|\Theta|=n+1$.
\end{lemma}

\begin{proof}
(1)\,\,This immediately follows from Remark 3.1 and Lemma 2.1.

(2)\,\,Assume that $h_{i}\geqslant2$ for all $i\in\Omega$. Then, for any $D\in\mathcal{I}(\mathbf{\overline{P}})$, there exists $\alpha\in\hat{\mathbf{H}}$ with $D=\langle\supp(\alpha)\rangle_{\mathbf{\overline{P}}}$. Now the result immediately follows from (1) and Corollary 2.1.
\end{proof}

\setlength{\parindent}{2em}
Now we are ready to prove Theorems 2.4 and 2.5.\\

\setlength{\parindent}{0em}
{\textbf{Proof of Theorem 2.4:}} Since $h_{i}\geqslant2$ for all $i\in\Omega$, $(1)\Longleftrightarrow(3)$ immediately follows from Lemma 3.4 and Theorem 3.5. Moreover, (2) implies that $|\Lambda|=n+1$, which, together with Corollary 2.1, immediately implies (1). Therefore it remains to prove $(1)\Longrightarrow(2)$. Now assume (1) holds. Then, by Lemma 3.4, we have $|\Theta|=n+1$. For $\alpha,\gamma\in\hat{\mathbf{H}}$ with $D=\langle \supp(\alpha)\rangle_{\mathbf{\overline{P}}}$, $B=\langle \supp(\gamma)\rangle_{\mathbf{\overline{P}}}$, by Lemma 3.4 and Theorem 3.5, we deduce that
$$\alpha\sim_{\Lambda}\gamma\Longleftrightarrow\pi(\Omega,D)=\pi(\Omega,B)\Longleftrightarrow|D|=|B|\Longleftrightarrow\wt_{\mathbf{\overline{P}}}(\alpha)=\wt_{\mathbf{\overline{P}}}(\gamma).$$
It then follows that $\Lambda=\mathcal{Q}(\hat{\mathbf{H}},\mathbf{\overline{P}})$, which further establishes (2), completing the proof.\\

{\textbf{Proof of Theorem 2.5:}} For $\alpha,\gamma\in\hat{\mathbf{H}}$, we let $D=\langle \supp(\alpha)\rangle_{\mathbf{\overline{P}}}$, $B=\langle \supp(\gamma)\rangle_{\mathbf{\overline{P}}}$. By Lemma 3.4 and Proposition 3.4, we deduce that $\alpha\sim_{\Lambda}\gamma$ if and only if $\pi(\Omega,D)=\pi(\Omega,B)$, if and only if there exists a bijection $\varepsilon:B\longrightarrow D$ such that $h_{i}=h_{\varepsilon(i)}$ for all $i\in B$. The above discussion immediately establishes Theorem 2.5, completing the proof.\\

\section*{Appendix}\appendix

\section{An association scheme approach to Conjecture 1.1}

\setlength{\parindent}{2em}
In this section, we give an alternative proof of Conjecture 1.1 using the notion of association scheme. We begin with the following definition, where we follow [3, Definition 2.1].

\setcounter{equation}{0}
\renewcommand\theequation{E.\arabic{equation}}

\setlength{\parindent}{0em}
\begin{definition}
Let $H$ be a nonempty finite set. Then, for $\Theta\subseteq 2^{H\times H}$, $(H,\Theta)$ is said to be an association scheme if the following five conditions hold:

{\bf{(1)}}\,\,$\Theta$ is a partition of $H\times H$;

{\bf{(2)}}\,\,$id_{H}\triangleq\{(x,x)\mid x\in H\}\in\Theta$;

{\bf{(3)}}\,\,For any $R\in\Theta$, $R^{-1}\triangleq\{(x,y)\mid (y,x)\in R\}\in\Theta$;

{\bf{(4)}}\,\,For any $R,S,T\in\Theta$ and for any $(u,v),(x,y)\in T$, it holds that
$$|\{z\mid z\in H,(u,z)\in R,(z,v)\in S\}|=|\{z\mid z\in H,(x,z)\in R,(z,y)\in S\}|;$$
{\bf{(5)}}\,\,For any $R,S\in\Theta$ and for any $(u,v)\in H\times H$, it holds that
$$|\{z\mid z\in H,(u,z)\in R,(z,v)\in S\}|=|\{w\mid w\in H,(u,w)\in S,(w,v)\in R\}|.$$
\end{definition}

\setlength{\parindent}{2em}
Fourier-reflexive partitions and association schemes are closely related. The following theorem has been established by Zinoviev and Ericson in [17, Theorem 1] (also see the discussion after [4, Theorem 2.4]). Here we state it in a form that is convenient for our discussion later.

\setlength{\parindent}{0em}
\begin{theorem}
Let $H$ be a finite abelian group, and let $\Gamma$ be a partition of $H$. Moreover, let $\Theta$ denote the partition of $H\times H$ such that
$$\text{$(x,y)\sim_{\Theta}(u,v)\Longleftrightarrow x^{-1}y\sim_{\Gamma}u^{-1}v$ for all $(x,y),(u,v)\in H\times H$}.$$
Then, $\Gamma$ is Fourier-reflexive if and only if $(H,\Theta)$ is an association scheme. Alternatively speaking, $\Gamma$ is Fourier-reflexive if and only if the following three conditions hold:

{\bf{(1)}}\,\,$\{1_{H}\}\in\Gamma$;

{\bf{(2)}}\,\,For any $B\in\Gamma$, $\{b^{-1}\mid b\in B\}\in\Gamma$;

{\bf{(3)}}\,\,For any $U,V,W\in\Gamma$ and for any $w,z\in W$, it holds that
$$|\{(u,v)\mid u\in U,v\in V,uv=w\}|=|\{(u,v)\mid u\in U,v\in V,uv=z\}|.$$
\end{theorem}

\setlength{\parindent}{0em}
\begin{remark}
In \cite{17}, the authors identify $\hat{H}$ with $H$ and use the notion of B-partitions, which are essentially Fourier-reflexive partitions.
\end{remark}

\setlength{\parindent}{2em}
Now we prove Conjecture 1.1. Our proof closely follows the spirits in the proofs of [11, Theorem 2.9] and [9, Theorem 3, part $\mathfrak{B}_3$], where the ambient space for codes are set to be $\mathbb{F}^{n}$ for some finite field $\mathbb{F}$ and $n\in\mathbb{Z}^{+}$. From now on until the end of this section, we let $\Omega$ be a finite set, and let $\left(H_{i}\mid i\in \Omega\right)$ be a family of finite abelian groups such that $|H_{i}|\geqslant2$ for all $i\in\Omega$. We set $\mathbf{H}=\prod_{i\in \Omega}H_{i}$, and we also fix a poset $\mathbf{P}=(\Omega,\preccurlyeq_{\mathbf{P}})$.\\

\setlength{\parindent}{0em}
{\textbf{Proof of Conjecture 1.1:}} Assume that $\mathcal{Q}(\mathbf{H},\mathbf{P})$ is Fourier-reflexive, and we show that $\mathbf{P}$ is hierarchical. Fixing $B\subseteq\Omega$ such that $a\preccurlyeq_{\mathbf{P}}b$ for all $a\in \Omega-B$, $b\in B$, we will show that
\begin{equation}\mbox{$c\preccurlyeq_{\mathbf{P}}d$ for all $c\in \min_{\mathbf{P}}(B)$, $d\in B-\min_{\mathbf{P}}(B)$}.\end{equation}
By way of contradiction, we assume that (E.1) does not hold. In the following, we will establish a series of facts step by step, and finally obtain a contradiction.

{\textbf{Step 1.}} Since (E.1) does not hold, we can choose $u\in\min_{\mathbf{P}}(B)$ and $v\in \min_{\mathbf{P}}(B-\min_{\mathbf{P}}(B))$ such that $u\not\preccurlyeq_{\mathbf{P}}v$. Since $|H_{v}|\geqslant2$, we can fix $\beta\in\mathbf{H}$ with $\supp(\beta)=\{v\}$. Moreover, it can be verified that there exists $k\in[2,|\min_{\mathbf{P}}(B)|]$ such that $\wt_{\mathbf{P}}(\beta)=|\Omega-B|+k$.

{\textbf{Step 2.}} For any $\alpha\in\mathbf{H}$ with $\wt_{\mathbf{P}}(\alpha)=|\Omega-B|+1$, it can be verified that $v\in\supp(\alpha^{-1}\beta)$, and hence $\wt_{\mathbf{P}}(\beta)\leqslant\wt_{\mathbf{P}}(\alpha^{-1}\beta)$. Therefore we have:
\begin{equation}\hspace*{-5mm}\{(\alpha,\eta)\in\mathbf{H}^{2}\mid\wt_{\mathbf{P}}(\alpha)=|\Omega-B|+1,\wt_{\mathbf{P}}(\eta)=\wt_{\mathbf{P}}(\beta)-1,\alpha\eta=\beta\}=\emptyset.\end{equation}

{\textbf{Step 3.}} Since $k\in[2,|\min_{\mathbf{P}}(B)|]$, there exists $M\subseteq\min_{\mathbf{P}}(B)$ such that $|M|=k$. Since $|H_{i}|\geqslant2$ for all $i\in\Omega$, we can choose $\theta\in \mathbf{H}$ such that $\supp(\theta)=M$. Moreover, it can be readily verified that
$$\wt_{\mathbf{P}}(\theta)=|\Omega-B|+k=\wt_{\mathbf{P}}(\beta).$$

{\textbf{Step 4.}} Fix $c\in M$. Then, there uniquely exists $\gamma\in\mathbf{H}$ such that $\supp(\gamma)=\{c\}$ and $\gamma_{(c)}=\theta_{(c)}$. It is straightforward to verify that $\wt_{\mathbf{P}}(\gamma)=|\Omega-B|+1$, $\wt_{\mathbf{P}}(\gamma^{-1}\theta)=\wt_{\mathbf{P}}(\theta)-1=\wt_{\mathbf{P}}(\beta)-1$, which, in particular, implies that
\begin{equation}\hspace*{-5mm}\{(\alpha,\eta)\in\mathbf{H}^{2}\mid\wt_{\mathbf{P}}(\alpha)=|\Omega-B|+1,\wt_{\mathbf{P}}(\eta)=\wt_{\mathbf{P}}(\beta)-1,\alpha\eta=\theta\}\neq\emptyset.\end{equation}

Since $\wt_{\mathbf{P}}(\theta)=\wt_{\mathbf{P}}(\beta)$, by (E.2), (E.3) and Theorem A.1, we deduce that $\mathcal{Q}(\mathbf{H},\mathbf{P})$ is not Fourier-reflexive, a contradiction, as desired. Therefore we conclude that (E.1) holds true. Finally, we apply Lemma 2.6 and reach $\mathbf{P}$ is hierarchical, completing the proof.

\end{document}